\documentclass[preprint, pra, amsmath]{revtex4}
\usepackage{enumerate}
\usepackage{graphicx}
\usepackage{amsmath, amsthm, amssymb}
\newtheorem{theorem}{Theorem}

\newtheorem{lemma}{Lemma}

\theoremstyle{remark}
\newtheorem*{remark}{Remark}
\usepackage{enumerate}
\usepackage{graphicx}
\usepackage{amsmath, amssymb}
\usepackage{graphics}
\begin{document}
\newcommand{\real}{\textrm{Re}\:}
\newcommand{\sto}{\stackrel{s}{\to}}
\newcommand{\Tr}{\textrm{Tr}\:}
\newcommand{\supp}{\textrm{supp}\:}
\newcommand{\wto}{\stackrel{w}{\to}}
\newcommand{\ssto}{\stackrel{s}{\to}}
\newcounter{foo}
\providecommand{\norm}[1]{\lVert#1\rVert}
\providecommand{\abs}[1]{\lvert#1\rvert}

\title{Zero Energy Bound States in Three--Particle Systems}

\author{Dmitry K. Gridnev}
\affiliation{FIAS, Ruth-Moufang-Stra{\ss}e 1, D--60438 Frankfurt am Main, Germany}
\altaffiliation[On leave from:  ]{ Institute of Physics, St. Petersburg State University, Ulyanovskaya 1, 198504 Russia}

\begin{abstract}
We consider a 3--body system in $\mathbb{R}^3$ with non--positive potentials and non--negative essential spectrum. Under certain  requirements on the fall off of pair potentials it is proved that if at least one pair of particles has a zero energy resonance then a square integrable zero energy ground state of three particles does not exist. This complements the analysis in \cite{1}, where it was demonstrated that square integrable zero energy ground states are possible given that in all two--body subsystems there is no negative energy bound states and no zero energy resonances. As a corollary it is proved that one can tune the coupling constants of pair potentials so that for any given $R, \epsilon >0$: (a) the bottom of the essential spectrum is at zero; (b) there is a negative energy ground state $\psi(\xi)$, where $\int |\psi(\xi)|^2 = 1$; (c) $\int_{|\xi| \leq R} |\psi(\xi)|^2 < \epsilon$.
\end{abstract}

\maketitle

\section{Introduction}\label{sec:1}

In \cite{1} under certain restrictions on pair potentials it was proved that the 3--body system, which is at the 3--body coupling constant threshold, has a square integrable state at zero energy if none of the 2--body subsystems is bound or has a zero energy resonance. Naturally, a question can be raised whether the condition on the absence of 2--body zero energy resonances is essential. Here we show that indeed it is. Namely, in this paper we prove that the 3--body ground state at zero energy can only be a resonance and not a $L^2$ state if at least one pair of particles has a resonance at zero energy. The method of proof is inspired by \cite{klaus1,sobol}. In the last section we demonstrate that there do exist 3--body systems, where (a) each 2--body subsystem is unbound; (b) one 2--body subsystem is at the coupling constant threshold; (c) the 3--body system has a resonance at zero energy and bound states with the energy less or equal to zero. System like this can be constructed through appropriate tuning of the coupling constants. 

\section{A Zero Energy Resonance in a 2--Body System}\label{sec:2}

In this section we shall use the method of \cite{klaus1} to prove a result similar to Lemma~2.2 in \cite{sobol}. Let us consider the Hamiltonian of 2 particles in $\mathbb{R}^3$
\begin{equation}\label{2ps}
h_{12}(\varepsilon) := - \Delta_x - (1+\varepsilon) V_{12} (\alpha x) ,
\end{equation}
where $\varepsilon \geq 0$ is a parameter and $\alpha= \hbar /\sqrt{2 \mu_{12}}$ is a constant, which depends on the reduced mass $\mu_{12}$ \cite{1}. Additionally, we require
\begin{list}{R\arabic{foo}}
{\usecounter{foo}
    \setlength{\rightmargin}{\leftmargin}}
\item
$0\leq V_{12} (\alpha x) \leq b_1 e^{-b_2 |x|}$, where $b_{1,2} >0$ are constants.
\item
$h_{12}(0) \geq 0$ and $\sigma_{pp}(h_{12}(\varepsilon)) \cap (-\infty, 0) \neq \emptyset$ for $\varepsilon >0$. \end{list}
The requirement R2 means that $h_{12}(0)$ has a resonance at zero energy, that is, negative energy bound states emerge iff the coupling constant is incremented by an arbitrary amount (in terminology of \cite{klaus1} the system is at the  coupling constant threshold).

The following integral operator appears in the Birmmann--Schwinger principle \cite{reed,klaus1}
\begin{equation}\label{bms}
    L(k) := \sqrt{V_{12}} \Bigl(- \Delta_x + k^2 \Bigr)^{-1} \sqrt{V_{12}}.
\end{equation}
$L(k)$ is analytic for $\real k >0$. Due to R1 one can use the integral representation and analytically continue $L(k)$ into the interior of the circle on the complex plane, which has its center at $k=0$ and the radius $|b_2|$ \cite{klaus1}. The analytic continuation is denoted as $\tilde L(k) = \sum_n \tilde L_n k^n$, where $\tilde L_n$ are Hilbert-Schmidt operators.
\begin{remark}
In Sec.~2 in \cite{klaus1} (page 255) Klaus and Simon consider only finite range potentials. In this case $L(k)$ can be analytically continued into the whole complex plane. As the authors mention it in Sec.~9 the case of potentials with an exponential fall off requires only a minor change: the analytic continuation takes place in a circle $|k| < b_2$.
\end{remark}

Under requirements R1-2 the operator $L(0) = \tilde L(0)$ is Hilbert-Schmidt and its maximal eigenvalue is equal to one
\begin{equation}
L(0) \phi_0 = \phi_0 . 
\end{equation}
$L(0)$ is positivity--preserving, hence, the maximal eigenvalue is non--degenerate and $\phi_0 \geq 0$. We choose the normalization, where $\|\phi_0 \| = 1$.

By the standard Kato--Rellich perturbation theory \cite{kato,reed} there exists $\rho > 0$ such that for $|k| \leq \rho$
\begin{equation}
\tilde L(k) \phi(k)= \mu(k)\phi(k) ,
\end{equation}
where $\mu(k), \phi(k)$ are analytic, $\mu(0) = 1$, $\phi(0)=\phi_0$ and the eigenvalue $\mu(k)$ is non--degenrate.
By Theorem 2.2 in \cite{klaus1}
\begin{equation}\label{muexp}
\mu (k) = 1 - a k + O(k^2) ,
\end{equation}
where
\begin{equation}\label{whatisa}
a = (\phi_0, (V_{12})^{1/2})^2/(4\pi) > 0 . 
\end{equation}
The orthonormal projection operators
\begin{gather}
\mathbb{P}(k) := (\phi(k), \cdot)\phi(k) = (\phi_0, \cdot)\phi_0 + \mathcal{O}(k) , \label{prp}\\
\mathbb{Q}(k) := 1 - \mathbb{P}(k) \label{prq}
\end{gather}
are analytic for $|k| < \rho$ as well. Our aim is to analyze the following operator function on $k \in (0, \infty)$
\begin{equation}\label{defw}
W(k) = [1-L(k)]^{-1} . 
\end{equation}
By the Birmann--Schwinger principle $\| L(k) \| < 1$ for $k >0$, which makes $W(k)$ well--defined.
\begin{lemma}\label{lem:1}
There exists $\rho_0 >0$ such that for $k \in (0,\rho_0)$
\begin{equation}\label{expans}
W(k) = \frac{\mathbb{P}_0}{ak} + \mathcal{Z}(k) ,
\end{equation}
where $\mathbb{P}_0 := (\phi_0 , \cdot)\phi_0 $ and $\sup_{k \in (0,\rho_0)} \|\mathcal{Z}(k)\| < \infty$.
\end{lemma}
\begin{proof}
$\tilde L(k) = L(k)$ when $k \in (0,\rho)$. We get from (\ref{defw})
\begin{gather}
W(k) = [1-L(k)]^{-1} = [1-L(k)]^{-1} \mathbb{P}(k)  +  [1-L(k)]^{-1} \mathbb{Q}(k) \\
= [1-\mu(k) \mathbb{P}(k) ]^{-1} \mathbb{P}(k) + [1-\mathbb{Q}(k) L(k)]^{-1} \mathbb{Q}(k) \\
= \frac 1{1-\mu(k)} \mathbb{P}(k) + \mathcal{\tilde Z}(k),
\end{gather}
where
\begin{equation}
\mathcal{\tilde Z}(k) := [1-\mathbb{Q}(k) L(k)]^{-1} \mathbb{Q}(k) . 
\end{equation}
Note that $\sup_{k \in (0,\rho)}\| \mathbb{Q}(k) L(k)\| < 1$  because the eigenvalue $\mu(k)$ remains isolated in this range. Thus $\mathcal{\tilde Z}(k) = O(1)$. Using (\ref{muexp}),(\ref{whatisa}) and (\ref{prp}) proves the lemma. \end{proof}

\begin{remark}
The singularity of $W(k)$ near $k=0$ has been analyzed in \cite{sobol} (Lemma 2.2 in \cite{sobol}), see also \cite{yafaev}. The decomposition (\ref{expans}) differs in the sense that $\mathcal{Z}(k)$ is uniformly bounded in the vicinity of $k=0$. The price we have to pay for that is the requirement R1 on the exponential fall off of $V_{12}$.
\end{remark}

\section{Main Result}\label{sec:3}

Let us consider the Schr\"odinger operator for three particles in $\mathbb{R}^3$
\begin{equation}\label{xc31aa}
    H = H_0 - V_{12}(r_1 - r_2) - V_{13}(r_1 - r_3) - V_{23}(r_2 - r_3),
\end{equation}
where $r_i$ are particle position vectors and $H_0$ is the kinetic energy operator with the center of mass removed.
Apart from R1-2 we shall need the following additional requirement
\begin{list}{R\arabic{foo}}
{\usecounter{foo}
    \setlength{\rightmargin}{\leftmargin}}
\setcounter{foo}{2}
\item$V_{13}, V_{23} \in L^2 (\mathbb{R}^3) + L^\infty_\infty (\mathbb{R}^3) $
 and  $V_{13}, V_{23} \geq 0$ and $V_{23} \neq 0$.
\end{list}
Here we shall prove
\begin{theorem}\label{th:1}
Suppose $H$ defined in (\ref{xc31aa}) satisfies R1-3. Suppose additionally that $H \geq 0$ and $H\psi_0 =0$, where $\psi_0 \in D(H_0)$. Then $\psi_0 = 0$.
\end{theorem}
We defer the proof to the end of this section. Our next aim would be to derive the inequality (\ref{e11})-(\ref{e12a}).

Like in \cite{1} we use the Jacobi coordinates $x = [\sqrt{2 \mu_{12}}/\hbar](r_2 - r_1)$ and $y = [\sqrt{2 M_{12}}/\hbar](r_3 - m_1/(m_1+m_2) r_1 - m_2/(m_1+m_2) r_2)$, where $\mu_{ij} = m_i m_j /(m_i + m_j)$ and $M_{ij} = (m_i + m_j)m_l / (m_i + m_j + m_l)$ are reduced masses (the indices $i,j,l$ are all different). The full set of coordinates in $\mathbb{R}^6$ is labeled by $\xi$. In the Jacobi coordinates the kinetic energy operator takes the form
\begin{equation}\label{ay4}
    H_0 = - \Delta_x - \Delta_y .
\end{equation}
Following the notation in \cite{1} $\mathcal{F}_{12}$ denotes the partial Fourier transform in $L^2(\mathbb{R}^6)$
\begin{equation}\label{ay5}
\hat f(x,p_y)  =  \mathcal{F}_{12} f(x,y) = \frac 1{(2 \pi )^{3/2}} \int d^3 y \; e^{-ip_y \cdot \; y} f(x,y) .
\end{equation}
(Here and always a hat over a function denotes its Fourier transform in $y$-coordinates). We shall need the following trivial technical lemmas. 
\begin{lemma}\label{lem:2}
Suppose an operator $A$ is positivity preserving and $\| A\| < 1$. Then  $(1-A)^{-1}$ is bounded and positivity preserving.
\end{lemma}
\begin{proof}
A simple expansion of $(1-A)^{-1}$ into von Neumann series. \end{proof}
\begin{lemma}\label{lem:3}
Suppose $g(y) \in L^2 \cap L^1 (\mathbb{R}^3) $  and $g(y) \geq 0, g \neq 0$. Then
\begin{equation}\label{sq1}
    \lim_{z \to + 0} \int_{|p_y|\leq \epsilon_0} d^3 p_y \frac{|\hat g|^2}{(p_y^2 + z^2)^{3/2}} = \infty
\end{equation}
for all $\epsilon_0 > 0$.
\end{lemma}
\begin{proof}
Let us set
\begin{equation}\label{sq3}
    J_{\epsilon}(z) = \int_{|p_y| < \epsilon} d^3 p_y \frac{1}{(p_y^2 + z^2)^{3/2}} \left| \int d^3 y e^{i p_y \cdot y} g(y) \right|^2 . 
\end{equation}
We have
\begin{equation}\label{sq4}
    J_{\epsilon} (z) \geq \int_{|p_y| < \epsilon} d^3 p_y \frac{1}{(p_y^2 + z^2)^{3/2}} \left| \int d^3 y \;  g(y) \cos{(p_y \cdot y)}\right|^2 . 
\end{equation}
Let us fix $r$ so that
\begin{equation}\label{sq4a}
    \int_{|y|>r} d^3 y g(y) = \frac 14 \|g\|_1
\end{equation}
Setting $\epsilon = \min [\epsilon_0 , \pi/(3r)]$ we get
\begin{equation}\label{sq5}
\cos{(p_y \cdot y)} \geq \frac 12  \quad \quad \mathrm{if} \quad |p_y| \leq \epsilon , |y|\leq r .
\end{equation}
Substituting (\ref{sq5}) and (\ref{sq4a}) into (\ref{sq4}) we get
\begin{equation}\label{sq6}
    J_\epsilon (z) \geq \frac{\| g \|_1^2}{64} \int_{|p_y| < \epsilon} d^3 p_y \frac{1}{(p_y^2 + z^2)^{3/2}} . 
\end{equation}
The integral in (\ref{sq6}) logarithmically diverges for $z \to +0$.  \end{proof}

\begin{remark}
Lemma~\ref{lem:2} may hold for $g(y) \in L^2 (\mathbb{R}^3)$ but we could not prove this. 
\end{remark}

So let us assume that there is a bound state $\psi_0 \in D(H_0)$ at zero energy, where $\psi_0 >0$ because it is the ground state \cite{reed}. Then we would have
\begin{equation}\label{e1}
    H_0 \psi_0 = V_{12} \psi_0 + V_{13} \psi_0 + V_{23} \psi_0 ,
\end{equation}
Adding the term $z^2 \psi_0$ (where here and further $z > 0$ ) and acting with an inverse operator on both sides of (\ref{e1}) gives
\begin{gather}
    \psi_0 = [H_0 + z^2]^{-1}V_{12} \psi_0 + [H_0 + z^2]^{-1} V_{13} \psi_0 + [H_0 + z^2]^{-1}V_{23} \psi_0 \label{e3}\\
    + z^2 [H_0 + z^2]^{-1}\psi_0 . \label{e3a} 
\end{gather}
From now we let $z$ vary in the interval $(0,\rho_0/2)$, where $\rho_0$ was defined in Lemma~\ref{lem:1}. Because the operator $[H_0 + z^2]^{-1}$ is positivity preserving \cite{reed} we obtain the inequality
\begin{equation}\label{e4}
    \psi_0 \geq [H_0 + z^2]^{-1}\sqrt{V_{12}} (\sqrt{V_{12}} \psi_0 )
\end{equation}
Now let us focus on the term $\sqrt{V_{12}} \psi_0 $. Using (\ref{e3}) we get
\begin{gather}\label{e5}
    \Bigl[1- \sqrt{V_{12}} (H_0 + z^2)^{-1} \sqrt{V_{12}}\Bigr]  \sqrt{V_{12}} \psi_0 =  \sqrt{V_{12}} [H_0 + z^2]^{-1} V_{13} \psi_0 \\
    +  \sqrt{V_{12}} [H_0 + z^2]^{-1}V_{23} \psi_0 + z^2  \sqrt{V_{12}} [H_0 + z^2]^{-1}\psi_0
\end{gather}
And by Lemma~\ref{lem:2}
\begin{equation}\label{e6}
    \sqrt{V_{12}} \psi_0 \geq  \Bigl[1- \sqrt{V_{12}} (H_0 + z^2)^{-1} \sqrt{V_{12}}\Bigr]^{-1}  \sqrt{V_{12}} [H_0 + z^2]^{-1} V_{23} \psi_0  . 
\end{equation}
It is technically convenient to cut off the wave function $\psi_0$ by introducing 
\begin{equation}\label{e7}
\psi_1 := \psi_0 (\xi) \chi_{\{\xi|\; |\xi| \leq b\}} ,
\end{equation}
where clearly $\psi_1 \in L^2 \cap L^1 (\mathbb{R}^6)$ and $b>0$ is fixed so that $\|V_{23} \psi_1\| \neq 0$ (which is always possible since $V_{23} \neq 0$).

Applying again Lemma~\ref{lem:2} we get out of (\ref{e6})
\begin{equation}\label{e6new}
    \sqrt{V_{12}} \psi_0 \geq  \Bigl[1- \sqrt{V_{12}} (H_0 + z^2)^{-1} \sqrt{V_{12}}\Bigr]^{-1}  \sqrt{V_{12}} [H_0 + 1]^{-1} V_{23} \psi_1 \: .
\end{equation}
Substituting (\ref{e6new}) into (\ref{e4}) gives that for all $z \in (0, \rho_0 /2)$
\begin{equation}\label{e11}
\psi_0 \geq f(z) \geq 0,
\end{equation}
where
\begin{gather}
    f(z) = [H_0 + z^2]^{-1}\sqrt{V_{12}} \Bigl[1- \sqrt{V_{12}} (H_0 + z^2)^{-1} \sqrt{V_{12}}\Bigr]^{-1}  \label{e12} \\
    \times \sqrt{V_{12}} [H_0 + 1]^{-1} V_{23} \psi_1 \: . \label{e12a}
\end{gather}

Our aim would be to prove that $\lim_{z \to +0} \|f(z)\| = \infty$, which would be in contradiction with (\ref{e11}).
Let us define
\begin{gather}\label{g(p_y)}
\Phi (x,y) := [H_0 + 1]^{-1} V_{23} \psi_1 , \\
    g(y) := \int d x \; \Phi (x,y) \sqrt{V_{12} (\alpha x)}\phi_0 (x) ,
\end{gather}
where $\phi_0$ is defined in Sec.~\ref{sec:2}.
\begin{lemma}\label{lem:4}
$g(y) \in L^1 \cap L^2 (\mathbb{R}^3)$ and $g(y) \neq 0$.
\end{lemma}
\begin{proof}
Following \cite{klaus1} let us denote by $G_0 (\xi - \xi', 1)$ the integral kernel of $[H_0 +1]^{-1}$. We need a rough upper bound on $G_0 (\xi,1)$. Using the formula on p. 262 in \cite{klaus1} we get
\begin{gather}
(4\pi)^3 |\xi|^4 e^{|\xi|/2} G_0 (\xi , 1) = \int_o^\infty t^{-3} e^{|\xi|/2} e^{-t|\xi|^2} e^{-1/(4t)} dt \\
\leq \int_0^\infty t^{-3} e^{-3/(16t)} dt = \frac{256}9
\end{gather}
Hence,
\begin{equation}\label{bgf}
G_0 (\xi,1) \leq \frac 4{9\pi |\xi|^4}  e^{-|\xi|/2}  . 
\end{equation}
Using $\|\sqrt{V_{12}}\phi_0\|_\infty < \infty $ we get $g(y) \in L^1 \cap L^2 (\mathbb{R}^3)$ if $\Phi \in L^1 \cap L^2 (\mathbb{R}^6)$. Because $\Phi \in L^2 (\mathbb{R}^6)$ to prove $\Phi \in L^1 (\mathbb{R}^6)$ it suffices to show that $\chi_{\{\xi|\; |\xi| \geq 2b\}}\Phi(\xi) \in L^1 (\mathbb{R}^6)$, where $b$ was defined after Eq.~(\ref{e7}). This follows from (\ref{bgf})
\begin{gather}\label{bgfa}
\chi_{\{\xi|\; |\xi| \geq 2b\}}\Phi(\xi) \leq \chi_{\{\xi|\; |\xi| \geq 2b\}} \int_{|\xi'|\leq b} d^6 \xi' G_0 (\xi - \xi',1) V_{23}\Psi_1(\xi') \\
\leq   \chi_{\{\xi|\; |\xi| \geq 2b\}} \frac 4{9\pi (|\xi| - b)^4}  e^{-(|\xi| - b)/2} \bigl\|V_{23}\Psi_1\bigr\|_1  \in L^1 (\mathbb{R}^6)
\end{gather}
That $g \neq 0$ follows from the inequality $\Phi (x,y) >0$. \end{proof}

Applying $\mathcal{F}_{12}$ to (\ref{e12})--(\ref{e12a}) we get 
\begin{gather}
    \hat f (z) = [-\Delta_x +p_y^2 + z^2]^{-1}\sqrt{V_{12}} \Bigl[1- \sqrt{V_{12}} (-\Delta_x +p_y^2 + z^2)^{-1} \sqrt{V_{12}}\Bigr]^{-1}   \\
    \sqrt{V_{12}} [-\Delta_x + p_y^2 + 1]^{-1} \widehat{V_{23} \psi_0}  . 
\end{gather}

By Lemma~\ref{lem:1} for $|p_y | < \rho_0 /2$ and $z < \rho_0 /2$
\begin{equation}\label{sobef}
    \Bigl[1- \sqrt{V_{12}} \Bigl(-\Delta_x +p_y^2 + z^2\Bigr)^{-1} \sqrt{V_{12}}\Bigr]^{-1} = \frac{\mathbb{P}_0}{a \sqrt{p_y^2 + z^2}} + \mathcal{Z}\Bigl(\sqrt{p_y^2 + z^2}\Bigr)  . 
\end{equation}
From now on $z\in (0,\rho_0/2)$. Notating shortly $\chi_0 (p_y) := \chi_{\{p_y | \; |p_y| < \rho_0 /2\}}$ gives us
\begin{equation}\label{e13}
    \chi_0 (p_y) \hat f(z) = \hat f_1 (z) + \hat f_2 (z),
\end{equation}
where
\begin{gather}
    \hat f_1(z) = \chi_0 (p_y) \frac{g(p_y)}{\sqrt{p_y^2 + z^2}} [-\Delta_x +p_y^2 + z^2]^{-1}\bigl(\sqrt{V_{12}} \varphi(x)\bigr) , \label{f1}\\
    \hat f_2(z) = \chi_0 (p_y) [-\Delta_x +p_y^2 + z^2]^{-1}\sqrt{V_{12}} \mathcal{Z}\Bigl(\sqrt{p_y^2 + z^2}\Bigr) \label{f2}\\
    \sqrt{V_{12}} [-\Delta_x + p_y^2 + 1]^{-1} \bigl(\mathcal{F}_{12}V_{23} \mathcal{F}^{-1}_{12}\bigr) \hat \psi_0 . \label{f2aa}
\end{gather}
The next lemma follows from the results of \cite{1}.
\begin{lemma}\label{lem:5}
$\sup_{z \in (0,\rho_0/2)}\|f_2 (z)\| < \infty$
\end{lemma}
\begin{proof}
Let us rewrite (\ref{f2})--(\ref{f2aa}) in the form
\begin{equation}\label{f2a}
f_2 (z) = \mathcal{A}(z) \mathcal{B}(z) \mathcal{C}(z) \psi_0 ,
\end{equation}
where
\begin{gather}
\mathcal{A}(z) = \chi_0 (p_y) [-\Delta_x +p_y^2 + z^2]^{-1}\sqrt{V_{12}} [1+t(p_y)+z]  , \\
\mathcal{B}(z) =  \chi_0 (p_y) \mathcal{Z}\Bigl(\sqrt{p_y^2 + z^2}\Bigr) , \\
\mathcal{C}(z) =   \chi_0 (p_y) \sqrt{V_{12}} [-\Delta_x + p_y^2 + 1]^{-1} [1+t(p_y)+z]^{-1} \bigl(\mathcal{F}_{12}V_{23} \mathcal{F}^{-1}_{12}\bigr) ,
\end{gather}
and  $t(p_y)$ is defined as in Eq.~(22) in \cite{1}. Note that by (\ref{expans}) $\tilde Z \Bigl(\sqrt{p_y^2 + z^2}\Bigr)$ is a difference of two operators each of which commutes with the operator of multiplication by $[1+t(p_y)+z]$. That $\sup_{z \in (0,\rho_0/2)}\|\mathcal{A} (z)\|, \|\mathcal{C} (z)\| < \infty$ follows directly from the proofs of Lemmas 6,8 in \cite{1}. $\sup_{z \in (0,\rho_0/2)}\|\mathcal{B} (z)\| < \infty$ follows from Lemma~\ref{lem:1}.\end{proof}

The last Lemma needed for the proof of Theorem~\ref{th:1} is
\begin{lemma}\label{lem:6}
$\lim_{z \to 0}\|f_1 (z)\| = \infty$.
\end{lemma}
\begin{proof}
We get
\begin{gather}
    \| \hat f_1(z) \|^2 = \frac 1{4\pi^2} \int_{|p_y| \leq \rho_0/2} d p_y \frac{|\hat g(p_y)|^2}{p_y^2 + z^2} \int dx \int dx' \int dx'' \frac {e^{-\sqrt{p_y^2 + z^2}|x-x'|}}{|x-x'|} \label{e15} \\
    \times \frac {e^{-\sqrt{p_y^2 + z^2}|x-x''|}}{|x-x''|} \bigl(\sqrt{V_{12}}(\alpha x') \varphi(x')\bigr) \bigl(\sqrt{V_{12}}(\alpha x'') \varphi(x'')\bigr)  . 
\end{gather}
The are constants $R_0, C_0 >0$ such that
\begin{equation}\label{expin}
\int d^3 x' \frac{e^{-\delta|x-x'|}}{|x-x'|} \sqrt{V_{12}(\alpha x')}\phi_0 (x') \geq C_0 \frac{e^{-2\delta|x|}}{|x|}\chi_{\{x|\; |x|\geq R_0\}}
\end{equation}
for all $\delta >0$. Indeed, the following inequality holds for all $R_0 >0$
\begin{equation}\label{zabyv}
\chi_{\{x|\; |x|\geq R_0\}} \frac{e^{-\delta|x-x'|}}{|x-x'|} \chi_{\{x|\; |x'|\leq R_0\}} \geq \frac{e^{-2\delta|x|}}{2|x|}\chi_{\{x|\; |x|\geq R_0\}}  . 
\end{equation}
Substituting (\ref{zabyv}) into the lhs of (\ref{expin}) we obtain (\ref{expin}), where
\begin{equation}
C_0 = \frac 12 \int_{|x'|\leq R_0} d^3 x' \sqrt{V_{12}(\alpha x')}\phi_0 (x')
\end{equation}
and one can always choose $R_0$ so that $C_0 >0$. Using (\ref{expin}) we get
\begin{equation}
 \| \hat f_1(z) \|^2  \geq c \int_{|p_y| \leq \frac{\rho_0}2} d p_y \frac{|\hat g(p_y)|^2}{(p_y^2 + z^2)^{3/2}} ,
\end{equation}
where $c>0$ is a constant. Now the result follows from Lemmas~\ref{lem:3}, \ref{lem:4}.  \end{proof}

The proof of Theorem~\ref{th:1} is now trivial.
\begin{proof}[Proof of Theorem~\ref{th:1}]
A bound state at threshold should it exist must satisfy inequality (\ref{e11}) for all $z \in (0,\rho_0/2)$. Thus $\|f(z)\|$ and, hence, $\|\chi_0 \hat f(z)\|$ are uniformly bounded for $z \in (0,\rho_0/2)$. By (\ref{e13}) and Lemmas~\ref{lem:5}, \ref{lem:6} this leads to a contradiction.  \end{proof}

\section{Example}\label{sec:4}

Suppose that R2 is fulfilled and let us introduce the coupling constants $\Theta , \Lambda >0$ in the following manner
\begin{equation}
H(\Theta, \Lambda) = [-\Delta_x - V_{12}] -\Delta_y - \Theta V_{13} - \Lambda V_{23}  . 
\end{equation}
For simplicity, let us require that $V_{ik} \geq 0$ and $V_{ik} \in C^\infty_0 (\mathbb{R}^3)$. Let $\Theta_{cr}, \Lambda_{cr}$ denote the 2--body coupling constant thresholds for particle pairs 1,3 and 2,3 respectively. On one hand, using a variational argument it is easy to show that there exists $\epsilon >0$ such that $H(\Theta, \Lambda) > 0$ if $\Theta, \Lambda \in [0,\epsilon]$ (that is in this range $H(\Theta, \Lambda)$ has neither negative energy bound states nor a zero energy resonance) \cite{gridnev,martin}. On the other hand, by the Efimov effect the negative spectrum of $H(\Theta_{cr}, \Lambda)$ is not empty for $\Lambda \in (0, \Lambda_{cr})$ \cite{sobol,yafaev}. So let us fix $\Lambda = \epsilon$ and let $\Theta$ vary in the range $[\epsilon, \Theta_{cr}]$. The energy of the ground state $E_{gr}(\Theta) = \inf \sigma \Bigl(H(\Theta, \epsilon)\Bigr)$ is a continuous function of $\Theta$. $E_{gr}(\Theta)$ decreases monotonically at the points where $E_{gr}(\Theta) <0$. Because $E_{gr}(\epsilon) = 0$ there must exist $\Theta_0 \in (\epsilon, \Theta_{cr})$ such that $E_{gr}(\Theta) <0$ for $\Theta \in (\Theta_0 , \Theta_{cr})$ and $E_{gr}(\Theta_0) = 0$.

Summarizing, $H(\Theta_0 , \epsilon)$ is at the 3--body coupling constant threshold. By Theorem~\ref{th:1} $H(\Theta_0 , \epsilon)$ has a zero energy resonance but not a zero energy bound state. If $\psi_{gr}(\Theta, \xi) \in L^2(\mathbb{R}^6)$ is a wave function of the ground state defined on the interval $(\Theta_0, \Theta_{cr})$ then for $\Theta \to \Theta_0 + 0$ the wave function must totally spread (see Sec.~2 in \cite{1}). Which means that for any $R>0$
\begin{equation}
\lim_{\Theta \to \Theta_0 + 0} \int_{|\xi| < R}  |\psi_{gr}(\Theta, \xi)|^2 d\xi \: \to 0 . 
\end{equation}

Note also that if the particles 1,2 would have a ground state at the energy $e_{12} < 0$ then the ground state of the 3 body system at the energy $e_{12}$ cannot be bound, it can only be a resonance \cite{klaus2}.

\begin{acknowledgements}
The author would like to thank Prof. Walter Greiner for the warm hospitality at FIAS.
\end{acknowledgements}

\end{document}